\documentclass{article}
\usepackage{aaai18}
\usepackage{graphicx}
\usepackage{amsmath,amssymb,amsthm}
\usepackage{comment}
\usepackage{color}
\usepackage{url}
\usepackage{courier}
\usepackage{helvet}
\DeclareMathOperator*{\argmin}{arg\,min}

\usepackage{times}
\usepackage[ruled,vlined,linesnumbered]{algorithm2e}

\newtheorem{theorem}{Theorem}
\newtheorem{assumption}{Assumption}
\newtheorem{lemma}{Lemma}
\newtheorem{definition}{Definition}
\newtheorem{corollary}{Corollary}

\newcommand*\samethanks[1][\value{footnote}]{\footnotemark[#1]}

\nocopyright

\title{Traffic Optimization For a Mixture of Self-interested and Compliant Agents}
\author{{Guni Sharon}\textsuperscript{1}\thanks{These authors contributed equally.},
{Michael Albert}\textsuperscript{2}\samethanks,
{Tarun Rambha}\textsuperscript{3}\samethanks,
{Stephen Boyles}\textsuperscript{1},
{Peter Stone}\textsuperscript{1}\\
\textsuperscript{1}{University of Texas at Austin}\\
\textsuperscript{2}{Duke University}\\
\textsuperscript{3}{Cornell University}\\
gunisharon@gmail.com,
malbert@cs.utexas.edu,
tr244@cornell.edu,
sboyles@mail.utexas.edu, pstone@cs.utexas.edu
}
\begin{document}

\maketitle

\begin{abstract}

  This paper focuses on two commonly used path assignment policies for agents traversing a congested network: \textit{self-interested routing}, and \textit{system-optimum routing}.
  In the self-interested routing policy each agent selects a path that optimizes its own utility, while the system-optimum routing agents are assigned paths with the goal of maximizing system performance.
  This paper considers a scenario where a centralized network manager wishes to optimize utilities over all agents, i.e., implement a system-optimum routing policy.
  In many real-life scenarios, however, the system manager is unable to influence the route assignment of all agents due to limited influence on route choice decisions.
  Motivated by such scenarios, a computationally tractable method is presented that computes the minimal amount of agents that the system manager needs to influence (compliant agents) in order to achieve system optimal performance.
  Moreover, this methodology can also determine whether a given set of compliant agents is sufficient to achieve system optimum and compute the optimal route assignment for the compliant agents to do so.
  Experimental results are presented showing that in several large-scale, realistic traffic networks optimal flow can be achieved with as low as 13\% of the agent being compliant and up to 54\%.

%We address the problem of routing flow through a congested network. The problem's input includes, the flow demand between any two vertices, and a latency function for each link which is a function of the volume of flow on it. The objective is to assign flows to paths such that the total latency is minimized. This paper views flow as a set of agents. The behavior assigned to each agent affects the total system latency. Previous work examined the total latency when all agents act \textit{selfishly}, i.e., minimize there own latency, or when all agents act \textit{selflessly}, i.e., minimize the total system-wide, latency. This work studies the system performance when there is a mixture of selfish and selfless agents, and particularly investigates the minimum number of selfless agents required for the system optimal state to obtain.  These results take the form of theoretical bounds, algorithms for computing this minimum number on a specific network, and experimental results from traffic simulators.
  
\end{abstract}

\section{Introduction} 

In multiagent systems, there are generally two paradigms of interaction.
Centralized control paradigms assume that a single decision making
entity is able to dictate the actions of all the agents, thus leading
them to a coordinated social optimum.
Decentralized control paradigms, on the other hand,
assume that each agent selects its own actions, and while it is in
principle possible for them to act altruistically, they are generally assumed to be
self-interested.
%A central theme of multiagent mechanism design is finding interaction mechanisms for self-interested agents that incentivize them to reach coordinated behavior that is as close as possible to the social optimum.

In this paper, we consider a routing scenario in which a subset of agents are
controlled centrally (\emph{compliant agents}), while the remaining are \emph{self-interested agents}.
We model the system as a Stackelberg routing game \cite{yang2007stackelberg} in which the decision maker for the centrally controlled agents is the leader, and the self-interested agents are the followers.
In this paper, we provide a computationally tractable methodology for 1) determining whether a given subset of centrally controlled agents are sufficient to achieve system optimum ($SO$), 2) determining the maximum number of agent that may be self-interested such that the centrally controlled agents can be deployed in order to induce $SO$, and 3) computing the actions the leader should prescribe to a sufficient set of compliant agents in order to achieve $SO$.

%As a concrete instance of this very general problem, in this paper we
%consider a traffic routing scenario in which some vehicles can be
%controlled centrally (for instance because they are part of a public
%transportation network, or because they have opted into an
%incentive-based routing scheme), and the rest are driven by individual,
%self-interested users.

It has been known for nearly a century in routing games that agents seeking to minimize their private latency need not minimize the total system's latency \cite{pigou2013economics,roughgarden2002bad}.
That is, self-interested agents may reach a user equilibrium ($U \! E$) that is not optimal from a system perspective.
However, if all agents are assigned paths with minimum system marginal cost then the system will achieve optimal performance~\cite{pigou2013economics,beckmann56,braess69}. 

Therefore, from a system manager perspective, it is desirable that all agents traversing a network would strictly utilize minimal marginal cost paths, even if the path is not a minimum latency path for an individual agent.
However, in many important scenarios, it will not be possible to enforce path assignment on all agents, but it may be possible to affect the behavior of a subset  (the compliant agents).
As a motivating example, consider an opt-in tolling system where drivers are given positive incentives to enroll but, in exchange, they will be subject to tolls that affect their route choice \cite{AAMAS17-delta}.
Another relevant example is virtual private network (VPN) path allocation.
While each packet within the VPN might be self-interested, a pro-social network manager might allocate virtual paths that are different from those preferred by the self-interested packets~\cite{fingerhut1997designing,duffield1999flexible}.

% We start by introducing the motivation behind our research.
% The problem definition is then presented along with the relevant terminology followed by a discussion of related work.
However, we show that, in the general case, computing the optimal assignment of compliant agents is NP-hard.
Therefore, we focus on the specific scenario where the portion of compliant agents is sufficiently large to achieve $SO$.
We present a novel \emph{linear program} ($LP$) representation for computing the maximal portion of self-interested agents that allow the system to achieve $SO$ and to determine whether a given set of compliant agents is sufficient to achieve $SO$.
Furthermore, we provide a method to tractably compute the flow assignment for the compliant agents such that $SO$ performance is guaranteed.%Furthermore, we prove that in this case, assigning minimal marginal cost paths to all compliant users leads to the system optimum.

We demonstrate, using a standard traffic simulator over a wide range of road networks, that the number of compliant agents necessary to achieve system optimum is a relatively small percentage of total flow (between 13\% and 53\%).

% where less than x\% \GS{TODO - update value} compliant agents are needed in order to reach system optimum in three(?) different real-life traffic scenarios.

\section{Motivation}  

 Recent advances in GPS based tolling technology \cite{numrich2012global} open the possibility of implementing micro-tolling systems in which specific tolls are charged for the use of every link within a road network.
 Setting tolls appropriately can influence self-interested drivers to prefer paths with minimum system marginal cost  and thus, lead to improved system performance \cite{AAMAS17-delta}.%Ideally such a system will be used to impose \textit{Marginal-Cost Tolls} which cause self-interested agents to traverse greedy-optimal paths. 
%It was previously proven that imposing Marginal-Cost Tolls on all drivers leads to optimized social welfare. 

 Unfortunately, political factors deter public officials from allowing such a micro-tolling scheme to be realized.
 Road pricing is known to cause a great deal of public unrest and is thus opposed by governmental institutions~\cite{schaller2010new}.
 To tackle this issue and avoid public unrest, it would be beneficial to have an \emph{opt-in} micro-tolling system where, given some initial monetary sign-up incentive, drivers choose to opt-in to the system and be charged for each journey they take based on their chosen route.
 The vehicles belonging to such drivers would need to be equipped with a GPS device as well as a computerized navigation system.
 Given the toll values and driver's value of time, the navigation system would suggest a minimal cost route where the cost is a function of the travel time and tolls. 

 While addressing the issue of political acceptance, an opt-in system  would result in traffic that is composed of a mixture of self-interested and compliant agents (compliant in the sense that the system manager can influence their route choice).
 Such a scenario raises some practical questions which are the focus of this paper,
namely, what portion of self-interested agents can the system tolerate while still reaching optimum performance? 
The answer to this question can help practitioners to determine both the level and the targeting of incentives in an opt-in system.
% Assuming that higher incentives would result in higher participation rates, this paper can help guide practitioners as to what levels of incentive should be offered to drivers opting-in to such a tolling system. 

\section{Problem definition and terminology}

The terminology in this paper follows that of Roughgarden and Tardos (\citeyear{roughgarden2002bad}).
We review the relevant concepts and notation in this section.

\subsection{The flow model}

The flow model in this work is composed of a directed graph $G(V,E)$, and a demand function $R(s,t) \rightarrow \mathbb{R}^{+}$ mapping a pair of vertices $s,t \in V$ to a non-negative real number representing the required amount of flow between source, $s$, and target, $t$.\footnote{The demand between any source and target, $R(s,t)$, can be viewed as an infinitely divisible set of agents (also known as a non-atomic flow~\cite{rosenthal1973class}).}
An instance of the flow model is a $\{G,R\}$ pair. 

$\mathcal{P}_{s,t}$ denotes the set of acyclic paths from $s$ to $t$.
Define $\mathcal{P}$ as the collection of all $\mathcal{P}_{s,t}$ (i.e., $\cup_{s,t \in V} \mathcal{P}_{s,t}$).
The variable $f_p$ represents the flow volume assigned to path $p$.
Similarly, $f_e$ is the flow volume assigned to link $e$.
By definition, the flow on each link ($f_e$) equals the summation of flows on all paths of which $e$ is a part.
Define the system flow vector as $f=\mathrm{vect}\{ f_p \}$.
$f$ is said to be \textit{feasible} if for all $s,t \in V$, $\sum_{p \in \mathcal{P}_{s,t}} f_p=R(s,t)$. 
%We say that a flow $f_p$ \textit{belongs} to a demand $r_i$ if $p \in \mathcal{P}_i$. %We say that an $x$ amount of flow is \textit{rerouted} if a flow $f_p$ is reduced by $x$ while a flow $f_{p'}$ is increased by $x$ and $f_p,f_{p'} \in \mathcal{P}_i$ for some $i$.

Each link $e \in E$ has a latency function $l_e(f_e)$ which, given a flow volume ($f_e$), returns the latency (travel time) on $e$.
Following Roughgarden and Tardos (\citeyear{roughgarden2002bad}) we make the following assumption:

\begin{assumption}
  \label{assum:latency_function}
The latency function $l_e(f_e)$ is non-negative,  differentiable, and non-decreasing for each link $e \in E$.
\end{assumption}

The latency of a simple path $p$ for a given flow $f$, is defined as $l_p(f) = \sum_{e \in p} l_e(f_e)$. 
A feasible flow $f$ is defined as a \textit{user equilibrium} ($U \! E$) if for every $s,t \in V$ and $p_a, p_b \in \mathcal{P}_{s,t}$ with $f_{p_a} > 0$ it holds that $l_{p_a}(f) \le l_{p_b}(f)$ (see Lemma 2.2 in \cite{roughgarden2002bad}). In other words, at $U \! E$, no amount of flow can be rerouted to a path with lower latency when the rest of the flow is fixed.
% This is also known as a \emph{Nash equilibrium} \cite{rosenthal1973class}.

Define the system cost associated with link $e$ as $c_e(f_e) = l_e(f_e)f_e$, the cost of a path $p$ as $c_p(f) = \sum_{e \in p} c_e(f_e)$ and the cost of a flow $f$ as $c(f) = \sum_{e \in E} c_e(f_e)$.
Define $c'_e(x) = \frac{d}{dx} c_e(x)$ and $c'_p(f) = \sum_{e \in p} c'_e(f_e)$.
A feasible flow $f$ is defined as a \textit{system optimum} ($SO$) flow if for every $s,t \in V$ and $p_a, p_b \in \mathcal{P}_{s,t}$ with $f_{p_a} > 0$ it holds that ${c'}_{p_a}(f) \le c'_{p_b}(f)$ (see Lemma 2.5 in \cite{roughgarden2002bad}).
In other words, at $SO$, the benefit from reducing the flow along any path is always less than or equal to the cost of by adding the same amount of flow to a parallel, alternative path. 
We follow Roughgarden and Tardos (\citeyear{roughgarden2002bad}), and make the following assumption:

\begin{assumption}
  \label{assum:cost_convex}
The cost function $c_e(f_e)$ is convex for each link $e \in E$.
\end{assumption}

Assumptions \ref{assum:latency_function} and \ref{assum:cost_convex} imply that the set of $SO$ flows correspond to the set of solutions of a convex program where the objective is to minimize $c(f) = \sum_{e \in E} c_e(f_e)$ (see Roughgarden and Tardos (\citeyear{roughgarden2002bad}) Corollary 2.7).

\subsection{Problem Definition}

%A flow $f$ can be viewed as an infinitely divisible set of agents (also known as non-atomic flow~\cite{rosenthal1973class}). In this representation, the agents choose for themselves which path to traverse and the flow assigned to each path is a result of the agents' decisions. We use this representation to define a scenario where the demand set $R$ is controlled by a mixture of self-interested, $R^s \subset R$, and compliant agents, $R^c \subset R$. Where all the demand is affiliated with exactly one behavior i.e., $R^s \cup R^c = R$ and $R^s \cap R^c = \emptyset$. The problems addressed in this paper are:

The focus of this paper is a scenario where the demand is partitioned into self-interested and compliant agents.
We define two types of controllers that assign paths to all of the agents.
These controllers are viewed as players in a Stackelberg game \cite{yang2007stackelberg}.

\begin{itemize}
\item $SO$-controller - Stackelberg leader, the $SO$-controller aspires to assign paths to the compliant subset of agents that, taking into account the self-interested agents' reaction, optimizes the systems performance (i.e. minimizes total latency).
  $SO$-controller assigned flow will be referred to as \emph{compliant flow}.
  % As we discuss, computing the optimal paths assignment for the $SO$-controller is NP-hard in the general case.
\item $U \! E$-controller - Stackelberg follower, considering the compliant agents' path assignment as fixed, the $U \! E$-controller assigns paths to the self-interested agents, the \emph{$U \! E$ flow}, such that a state of user equilibrium (as defined above) is achieved.\footnote{The $U \! E$ enforced by the $U \! E$-controller applies only for the self-interested subset of agents.
    That is, no \textbf{self-interested} agent can benefit from unilaterally deviating from its assigned path.} 
\end{itemize}

\noindent The problems addressed in this paper are:

\begin{enumerate}
\item Given an instance of the flow model $\{G,R\}$, what is the maximum amount of self-interested agents that can be assigned to the $U \! E$ controller and still permit the $SO$ controller to achieve system optimum? 
\item Given a set of compliant agents and an instance of the flow model $\{G,R\}$, can the $SO$ controller assign paths to them in such a way that the system achieves $SO$?
\item If $SO$ is achievable, how should the $SO$-controller assign the compliant flow? Equivalently, what is the optimal Stackelberg equilibrium?
\end{enumerate}

\noindent To the best of our knowledge, this work is the first to answer these questions in a general setting.

\section{Related Work}

%Previous work \cite{roughgarden2002bad} examined the worst-case ratio of total latency at $U \! E$ to the same instance at SO, that is, the maximal value of $c(\mbox{U \! E})/c(\mbox{SO})$.  It was shown that, in the worst case, $c(\mbox{U \! E})$ for a given instance $(G,R)$ equals $c(\mbox{SO})$ for the same instance with double the demand $R'=  \cup_{r \in R} 2r$. Nonetheless, $c(\mbox{U \! E})/c(\mbox{SO})$ was shown to be unbounded in the worst case, although finite bounds can be given if latency functions are restricted to particular classes.

%Previous work examined mixed equilibrium scenarios where traffic is composed of: self-interested, compliant and Cournot-Nash (CN) agents. Compliant agents are assumed to choose routes that minimize the total system's latency. By contrast, Cournot-Nash (CN) agents belonging to the same group 

Previous work examined mixed equilibrium scenarios where traffic is composed of: $U \! E$ and Cournot-Nash ($C \! N$) controllers.
A $C \! N$-controller assigns flows to a given subset of the demand with the aim of minimizing the total travel time only for that subset.
For instance, a logistic company with many trucks can be viewed as a $C \! N$-controller. %If there is a single CN player controlling all the demand then the solution correspond to the $SO$ solution. On the other hand, if there are an infinite amount of infinitesimal CN agents the solution correspond with the $U \! E$ equilibrium. 

It was shown that the equilibrium for a mixed $U \! E$, $C \! N$ scenario is unique and can be computed using a convex program~\cite{haurie1985relationship,yang2008existence}.
On the other hand, no tractable algorithm is known for computing the optimal Stackelberg equilibrium for scenarios that also include a $SO$-controller.
% The last statement is not surprising since, as discuss, finding such an equilibrium is NP-hard.

Korilis et. al. (\citeyear{korilis1997achieving}) examined mixed equilibrium scenarios that do include a $SO$-controller.
In their work, a technique for computing the a solution for the above questions \#1 and \#3 (see problem definition) was suggested for specific types of flow models.
Their technique was proven to work for networks with a common source and a common target with any number of parallel links.
Moreover, the latency functions were assumed to be of a very specific form (linear function with a capacity bound).
As a result, their solution is not applicable when general networks with arbitrary latency functions are considered. 

Other work \cite{roughgarden2004stackelberg,immorlica2009coordination} studied a variant of the scheduling problem where infinitesimal jobs must be assigned to a set of shared machines each of which is affiliated with a non-negative, differentiable, and non-decreasing latency function that, given the machine load, specify the amount of time needed to complete a job.
When considering a scenario where part of the jobs are assigned to machines by a $U \! E$-controller while the rest are assigned by a $SO$-controller, they show it is NP-hard to compute the optimal Stackelberg equilibrium~\cite{roughgarden2004stackelberg}.
%\MA{I took a section out here about an approximation algorithm. I left it commented if you disagree and need to put it back in.}
% \MA{I think we should consider taking out the discussion on the approximation scheme. I think it weakens our argument for our methodology, and I think that the related work section may be too long.}
% As a result, the same work also suggests a tractable, sub-optimal, algorithm for computing the $SO$-controller's assignments (the Stackelberg leader strategy).
% The resulting Stackelberg equilibrium is guaranteed to induce total latency that is not more than $\lvert r \lvert / \lvert r_{SO} \lvert$, where $\lvert r \lvert$ is the total demand volume and $\lvert r_{SO} \lvert$ is the volume of compliant demand.
Their problem can be viewed as a special case of our problem, specifically a network with a single source and target with multiple parallel links between them.
Given that in this more restrictive setting computing the optimal Stackelberg equilibrium is intractable, the general question in our setting will also be computationally intractable.
% These findings are very relevant to our paper as the presented scheduling problem variant can be formulated as a traffic flow assignment problem where there is a single source, a single target, and each machine is a parallel link leading from source to target.
% Such a polynomial reduction implies that computing the Stackelberg equilibrium in the flow model defined in our paper is also NP-hard.
% Note however, that the bounded sub-optimal algorithm presented for the scheduling  problem variant cannot be applied to our, general, flow model since our model allows networks with many \{source, target\} pairs and is not constrained to accommodate only parallel links.

\section{Computing the Maximal $U \! E$ Flow}

% \MA{Need to state that this is for striclty \emph{convex cost functions}, not strictly increasing latency functions. EDIT: I was wrong about this, strictly increasing latency functions are enough.}

Given that finding the optimal Stackelberg equilibrium is NP-hard for an arbitrary number of compliant agents, this work focuses on scenarios where the number of compliant agents is sufficient to achieve $SO$.
As we will show, finding the optimal Stackelberg equilibrium can be done in polynomial time for such cases.
In this section, we will present a computationally tractable method to compute the maximal $U \! E$ flow given an instance of a flow model $\{G,R\}$, and we will provide a method to check, for a given level of compliant flow, whether $SO$ is achievable.

% However, we must first address the question: What is a sufficient number of compliant agents, and more generally what is the maximal amount of demand comprised of self-interested agents (routed by the $U \! E$-controller) that a problem instance $\{G,R\}$ can tolerate and still reach SO? i.e., what is the value of $r^*_{U \! E}$?
% Alternatively, if we have a given set of compliant agents assigned to the $SO$-controller, can we be guaranteed to reach $SO$?

We define $r^*_{U \! E}$ as the maximal amount of demand comprised of self-interested agents that the system can tolerate and still achieve $SO$.
Additionally, we define $r^*_{s,t}$ as the amount of demand from source $s$ to target $t$ that is assigned to the $U \! E$-controller.
That is, computing $r^*_{U \! E}$ is equivalent to maximizing $\sum_{s,t} r^*_{s,t}$.

We can cast the problem of maximizing $\sum_{s,t} r^*_{s,t}$ as an optimization problem, specifically a \emph{linear program} ($LP$).
% Now, we can view our problem (computing $r^*_{U \! E}$) as a linear program ($LP$) optimization problem where the objective function is:
% \begin{equation}
% max ~ \sum_{st \in V^2} r^*_{s,t} \label{eq:1}
% \end{equation}
Assigning values to all variables of type $r^*_{s,t}$ must follow some constraints.
Specifically, the flow from each origin to each destination must be both a \emph{subflow} of some $SO$ flow, and it must be an acceptable path for all self-interested agents in the flow, given the compliant flow.
% that are imposed by the flow model and the $U \! E$-controller's strategy. Listing these constraints requires us to first define several new properties, variables and constants.

% \MA{I want to introduce subflows here.}

\begin{definition}[Subflow of flow $f$]
  For a directed graph $G(V,E)$ and demand function $R$, a flow $f^*$ is a \emph{subflow of flow $f$} if for all links $e \in E$, $0 \le f^*_e \le f_e$ and for each pair of nodes $s,t \in V$, there exists $0 \le r_{s,t} \le R(s,t)$ such that
  \begin{equation*}
  \sum_{e \in out(s)} f^*_e -\sum_{e \in in(s)} f^*_e = \sum_t r_{s,t} \label{eq:demand_constraint}
\end{equation*}
and
  \begin{equation*}
  \sum_{e \in in(t)} f^*_e - \sum_{e \in out(t)} f^*_e = \sum_s r_{s,t} \label{eq:demand_constraint2}.
\end{equation*}

\end{definition}

A path will only be acceptable for a self-interested agent if it is the \emph{lowest latency path} from the origin to the destination given the compliant flow, and in order for the path to be part of a valid $SO$ flow, it must be a \emph{minimum marginal cost path}.
Therefore, a path $p$, leading from vertex $s$ to vertex $t$, will be said to be \textit{zero reduced cost} if there is no other path, $p'$, leading from $s$ to $t$ with lower latency or lower marginal cost. 

\begin{definition}[Zero reduced cost path] \label{def:zero_reduced}
  For a flow model $\{G,R\}$, a \emph{zero reduced cost path} with regard to flow assignment $f$ is a path $p \in \mathcal{P}_{s,t}$ such that $\forall \, p' \in \mathcal{P}_{s,t} : l_{p}(f) \le l_{p'}(f) ~\text{and}~ c'_{p}(f) \le c'_{p'}(f)$.
  A link, $e$, is defined as a \emph{zero reduced cost link} with respect to source $s$ if it is part of any zero reduced cost path originating from $s$ and terminating at $t$ for some origin-destination pair $(s,t) \in V^2$.
  We denote the set of zero reduced cost links with respect to source $s$ as $E^s_{RC}$
\end{definition}

% Formally, $p$ is a zero reduced cost path with regard to flow assignment $f$ if $\forall p' \in \mathcal{P}_{s,t} : l_{p}(f) \le l_{p'}(f) ~\text{and}~ c'_{p}(f) \le c'_{p'}(f)$. A link, $e \in E$, is defined as a zero reduced cost link with respect to source $s$ if it is part of any zero reduced cost path originating from $s$.
We require that the $U \! E$ flow (flow routed by the $U \! E$-controller) is routed solely via zero reduced cost links/paths.
This is because the $U \! E$ controller can only assign flow to minimal latency paths (otherwise self-interested agents would deviate).
However, the need to constrain $U \! E$ flow to links/paths with minimal marginal cost ($c'$) is less intuitive; this constraint will be justified later on.
Note that it is sufficient to only consider whether or not a link $e$ is part of a reduced cost path from the origin $s$ to \emph{some} destination $t$ (not a specific $t$) because either link $e$ is along a reduced cost path from $(s,t)$, or there is no path only along links in $E^s_{RC}$ that includes $e$.
Moreover, we can efficiently compute the set of zero reduced cost links for any origin destination pair $(s,t)$ by applying uniform cost search from $s$ to $t$ and marking all links that are part of optimal paths, once with regard to minimal total latency ($\argmin_{p \in \mathcal{P}_{s,t}}(l_p(f^{SO}))$, and second with regard to minimal marginal cost ($\argmin_{p \in \mathcal{P}_{s,t}}(c'_p(f^{SO}))$.
% Dijkstra's algorithm from the origin to the beginning node in the link and then again from the ending node in the link to the destination, where the length of a link is either the latency at the $SO$ solution (to check for minimum latency) or the marginal cost at the $SO$ solution (to check for minimum marginal cost).
% \MA{Need to justify why it is the case that if a link is part of a zero reduced cost path from an origin, then there is either no path in the zero reduced cost links that connects that link to the destination. Basically, we need to justify why we only have to look at zero reduced cost paths with respect to the source, and not both the source and the destination.}

%For each link, $e \in E$, and source vertex, $s$, let $a^s_{e}$ equal $1$ if $e$ is a zero reduced cost link with respect to $s$ else $a^s_e = 0$. Such variables are used to enforce the zero reduced cost constraint. 

%For each vertex pair $v1,v2 \in V^2$ and source vertex $s$ let $a^s_{v1,v2}$ and $a^s_{v2,v1}$ equal $+1$ and $-1$ respectively if there exists a zero reduced cost edge with regard to $s$ leading from $v1$ to $v2$. Else $a^i_{v1,v2} = a^i_{v2,v1} = 0$. This set of variables is used to enforce the zero reduced cost constraint. Let $A^s$ be a matrix (dimension =  $\lvert V \lvert \times \lvert V \lvert$) with one entry per vertex pair.

Let the constant $f^{SO}$ denote the flow vector at a $SO$ solution.\footnote{A $SO$ flow can be efficiently computed as a solution to a convex program \cite{roughgarden2002bad,dial2006path}.}
The $SO$ flow is not unique when latency functions are non-decreasing, and the maximal amount of $U \! E$ flow permitted may, in general, depend on the specific $SO$ flow.
Therefore, we must efficiently search over the space of $SO$ flows.
This is possible due to the following lemmas.

\begin{lemma}
  \label{lemma:constant_latency_SO}
  For any two flows that achieve SO, $f^{SO}$ and $\hat{f}^{SO}$, $l_e(f^{SO}_e) = l_e(\hat{f}_e^{SO})$.
  % \GS{Overloading of the term $f^*_e$. (We will also refer to the $U \! E$ linear program solution defined flow $f^*_e = \sum_v x^v_e\cdot a_e^s$ as the \emph{$U \! E$ flow assignment}.)}
\end{lemma}

\begin{proof}
Given Assumption~\ref{assum:cost_convex}, a $SO$ flow is the solution to a convex program \cite{roughgarden2002bad}. The solutions to a convex program form a convex set. Suppose that there are two flows that both achieve SO, but for which $f_e^{SO} \ne \hat{f}_e^{SO}$. Then $c_e(f_e) = l_e(f_e)f_e$ must be a linear function between $f_e^{SO}$ and $\hat{f}_e^{SO}$ (to see this, note that any convex combination of $f^{SO}$ and $\hat{f}^{SO}$ is also an $SO$ solution, but if $c_e(f_e)$ is not linear, then the total system travel time would be strictly less, a contradiction). Since $l_e(f_e)$ is a non-decreasing function, the only way for $c_e(f_e)$ to be linear is for $l_e(f_e)$ to be constant between $f_e^{SO}$ and $\hat{f}_e^{SO}$.
\end{proof}

\begin{lemma}
  \label{lemma:identical_reduced_cost_paths}
  The set of zero reduced cost paths is identical for all $SO$ solutions.
\end{lemma}

\begin{proof}
  By Lemma~\ref{lemma:constant_latency_SO}, all $SO$ flows have the same latency on each link, so the $SO$ solutions can differ by at most flows along a set of links with constant latency over the range of which the two flows differ on those links.
  Since we assume that the latency functions are differentiable, the derivatives of the latency function are zero over the range at which they are constant. Therefore, $c_e'(f_e) = l_e(f_e) + f_e l'_e(f_e)$ is constant over the range as well.
  This implies that any path that is reduced cost in one flow is also reduced cost in the other flow, since the latency functions and $c_e'(f_e)$ are constant for every link $e$.
\end{proof}

% We are now ready to introduce the modification to the $U \! E$ linear program.
% We will follow the notation from the previous discussion on computing the $U \! E$ flow assignment with strictly increasing latency functions, with a minor modification.

Define the constant $\bar{f}_e^{SO} = \mathrm{sup}\{f : l_e(f) = l_e(f^{SO}_e)\}$, i.e. $\bar{f}_e^{SO}$ is the largest flow value such that the latency on link $e$ is equal to the latency at an $SO$ solution.
Note that if $l_e$ is strictly increasing at $f^{SO}_e$, then $\bar{f}_e^{SO} = f^{SO}_e$.
However, if $l_e$ is constant at $f^{SO}_e$, then $\bar{f}_e^{SO} > f^{SO}_e$.

Given that the zero reduced cost paths are the same for \emph{all} $SO$ flows (Lemma~\ref{lemma:identical_reduced_cost_paths}), and any $SO$ flow has the same latency on all links (Lemma~\ref{lemma:constant_latency_SO}), it will be sufficient to only search over flows that are less than $\bar{f}_e^{SO}$ on each link $e \in E$.

% Therefore, we replace constraint~\eqref{eq:bundle} with the following constraint:
% \begin{equation}
%   \label{eq:mod_bundle}
%   \sum_{s} x_{e}^s \leq \bar{f}_e^{SO} \qquad  \forall \, e \in E, \,  \, s \in V   
% \end{equation}

For each vertex, $s$, and link, $e$, define variable $x^s_{e}$ denoting the amount of $U \! E$ flow originating from source $s$ that is assigned to link $e$. 
Let $in(v)$ denote the set of links for which $v$ is the tail vertex and $out(v)$ the set of links for which $v$ is the head vertex. 
% Now, the constraints on the objective function can be defined as:
% \MA{We need to do a better job of indicating what is a variable and what is a constant, and I think we should probably put the objective with the constraints.} 

%Let $x^s_{v1,v2}$ denote the amount of $U \! E$ flow originating from source $s$ that is assigned to the edge connecting $v1$ with $v2$. Let $X^s$ be a vector (dimension =  $\lvert V \lvert \times \lvert V \lvert$) with one entry per vertex pair. Finally, let $f_e$ denote the flow on edge $e$ at the $SO$ solution.\footnote{Computing the $SO$ for a given flow instance can be computed in polynomial time using ... \GS{TODO - cite algorithm}} Now, the constraints on the objective function can be defined as:

\begin{definition}
\label{definition:UE_LP}
For a given flow model $\{G, R\}$, the \emph{$U \! E$ linear program} is:

\begin{equation}
\max_{r_{s,t}^*, x_e^s} ~ \sum_{s,t \in V} r^*_{s,t}  \label{eq:1}
\end{equation}
\begin{align}
&\textit{subject to} && \nonumber\\
&r^*_{s,t} \leq R(s,t) &&  \forall \, s,t \in V \label{eq:demandcon} \\ 
&\sum_{e \in out(s)} x^s_e  = \sum_{t \in V} r^*_{s,t}  && \forall \,s \in V  \label{eq:flowcon_s} \\ 
&\sum_{e \in in(t)} \! x^s_e - \sum_{e \in out(t)} \! x^s_e = r^*_{s,v}  &&  \forall \, s,  t \in V \label{eq:flowcon_t} \\ 
&\sum_{s} x_{e}^s \leq \bar{f}^{SO}_e && \forall \, e \in E, \,  \, s \in V  \label{eq:bundle}\\
&x_{e}^s \geq 0 ,~ r^*_{s,t} \geq 0  &&\forall \, s, t \in V, \, e \in E  \label{eq:nonnegative} \\
&x_{e}^s = 0 &&\forall \, s \in V, \, e \in E \setminus E^s_{RC}  \label{eq:non_reduced_zero_flow} 
\end{align}
The flow $f^{U \! E}_e = \sum_v x^v_e$ defined by a feasible solution to the $U \! E$ linear program (given constraints \eqref{eq:demandcon}-\eqref{eq:non_reduced_zero_flow}) is a \emph{$U \! E$ subflow}.
The flow defined by an optimal solution to the $U \! E$ linear program is an \emph{optimal $U \! E$ subflow}.
\end{definition}

% We will refer to the linear program defined by equation~\eqref{eq:1}-\eqref{eq:nonnegative} as the \emph{$U \! E$ linear program}.
Note that the number of variables is $|\{\forall s \in V, \, \forall t \in V, \, \forall e \in E : r^*_{s,t}, \, x^s_e \}| = O(|V|^2 + |V||E|)$, and  the number of constraints is also $O(|V|^2 + |V||E|)$.
Therefore, since the number of variables and constraints are polynomial in the flow model, the optimal solution to the $U \! E$ linear program can be computed in polynomial time \cite{Karmarkar-1984}.
% We will also refer to the $U \! E$ linear program solution defined flow $f^*_e = \sum_v x^v_e$ as the \emph{$U \! E$ flow assignment}.

\begin{figure}[tbp]
\centering
%\vspace{-0.2 cm}
\includegraphics[width=\columnwidth]{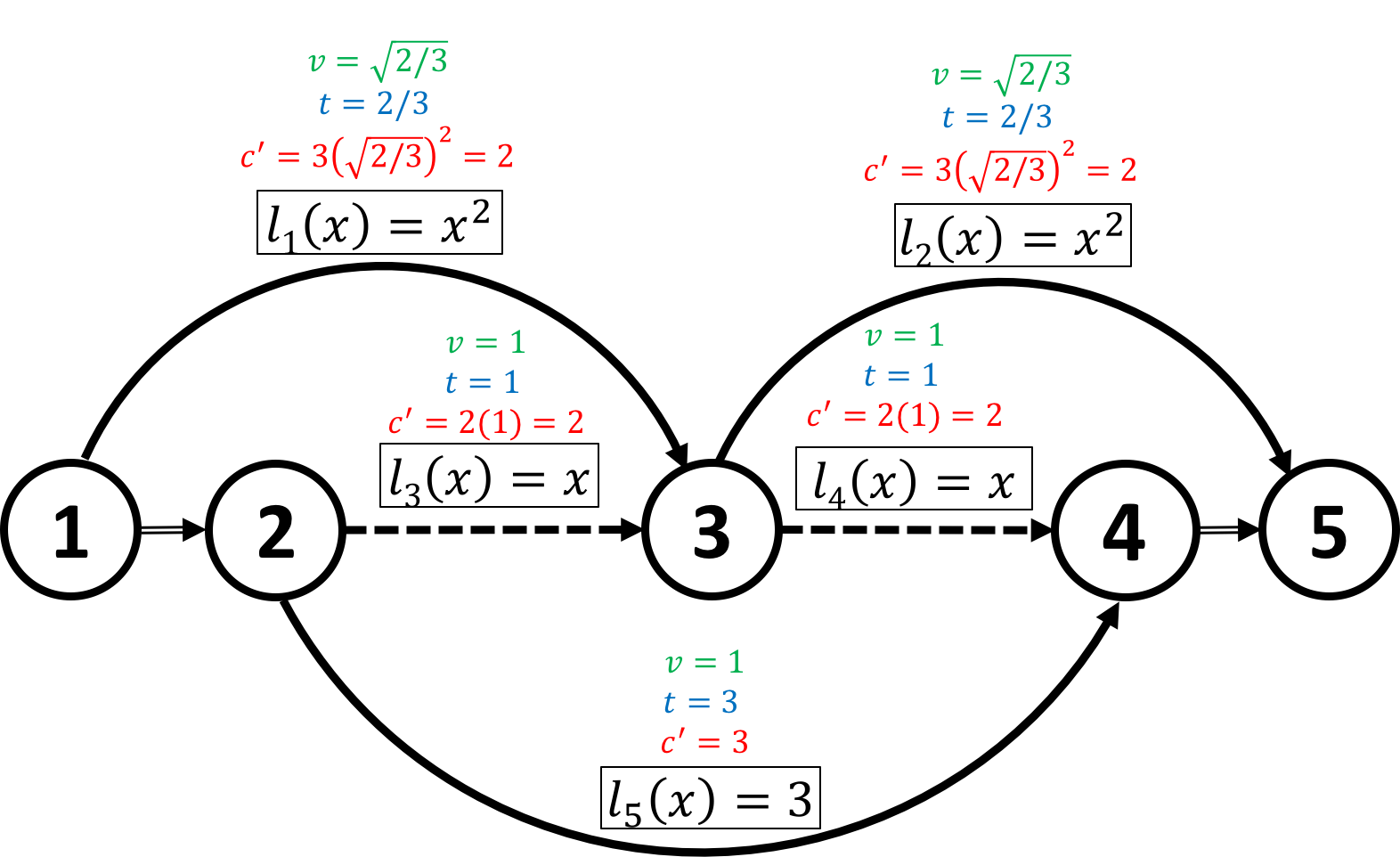}
\caption{A problem instance where the minimal marginal cost condition is required. $R(1,3)=R(3,5)=1+\sqrt{2/3},~R(2,4)=1$. Above each link are listed (in top to down order): $v$, the flow at $SO$ (in green), $t$, the travel time at $SO$ (in blue), $c'$, the marginal cost at $SO$ (in red), and $l(x)$, the latency function (in a bounding box).}%\vspace{-0.2 cm}
\label{fig:marg_cost_req}
% \label{fig:badc}
\end{figure}

Revisiting the definition of a zero reduced cost path/link (Definition \ref{def:zero_reduced}), the need to constrain $U \! E$ flow to  minimal marginal cost paths (in addition to minimal latency paths) is explained by the fact that, at $SO$, the $U \! E$ flow must be a subflow of $f^{SO}$.
Consider the problem instance depicted in Figure \ref{fig:marg_cost_req}.
The latency function, flow at SO, travel time at SO, and marginal cost at $SO$ are all listed above each link.
The double line links have zero travel time regardless of the volume, their only purpose is to limit the possible path assignment between vertices $2$ and $4$.
Notice that, for flow traveling from vertex $2$ to $4$, the dotted path is of minimal latency but not minimal marginal cost.
Running the above $LP$ on this instance would result in $r^*_{U \! E} = 2 \cdot \sqrt{2/3}$ which is the correct value.
However, if the dotted path is considered to be of zero reduced cost with respect to source 2 (despite having a non minimal marginal cost) then running the $LP$ would result in $r^*_{U \! E} = 1 + 2 \cdot \sqrt{2/3}$. 
% Adding the minimal marginal cost constraint limits the $U \! E$-flow belonging to a given $(s,t)$ pair on a given path by the $SO$-flow belonging to the same $(s,t)$ pair on the same path. 

% \MA{I don't think we should bother with any proofs in this section. We are going to prove everything in the harder case anyway, so I think it's best just to give intuition for the methodology.}

\begin{theorem}
\label{theorem:UE_flow_is_subflow}
  A $U \! E$ subflow, $f^{U \! E}$, defined by a feasible solution to the $U \! E$ linear program is a subflow of a $SO$ flow.
\end{theorem}

\begin{proof}
First, note that by equations~\eqref{eq:demandcon}--\eqref{eq:flowcon_t}, the subflow $U \! E$ subflow, $f^{U \! E}_e$, satisfies flow conservation constraints.
Equation~\eqref{eq:demandcon} states that the flow along all reduced cost paths from origin $s$ to destination $t$ must be less then total demand for $(s,t)$.
Then equations~\eqref{eq:flowcon_s} and \eqref{eq:flowcon_t} state that the flow out of node $v$ must either be due to the demand generated by node $v$ or the flow into it, minus the flow that reaches $v$ as a destination.
Therefore, $f^{U \! E}_e$ is a subflow of a feasible flow.

What must be shown is that there must exist a $SO$ flow, $f^{SO}$, such that $f^{U \! E}_e \le f^{SO}_e$ for all $e$.
If $e$ is such that $l_e$ is strictly increasing at an $SO$ solution, and therefore will be strictly increasing at all $SO$ solutions by Lemma~\ref{lemma:constant_latency_SO}, then $f^{SO}_e = \bar{f}_e^{SO}$ and constraint \eqref{eq:bundle} guarantees the claim.
Let $E'$ be the set of links such that the latency function is constant at a $SO$ flow.
% , and let $V'$ be the set of nodes such that for all $e \in E'$ where $e$ is from node $v$ to node $w$, $v, w \in V'$.
% Let $E^* = E \setminus E'$, and define the flow into and out of the subgraph defined by $(E',V')$ for the subflow, $f^*$ be $d^*_v = \sum_{e \in out(v), e \in E^*} f^*_e - \sum_{e \in in(v), e \in E^*} f^*_e$
Therefore, it only needs to be shown that there exists a $SO$ solution, $f$, such that for $e \in E'$, $f^{U \! E}_e \le f^{SO}_e$.

Suppose that there existed a set of links $e \in E'$ such that for all $SO$ flows $f^{SO}$, $f^{U \! E}_e > f^{SO}_e$.
Let $\hat{f}^{SO}$ be an $SO$ flow.
Then there must exist an origin destination pair $(s,t)$ such that there are two sets of paths $\mathcal{P}_{>}, \mathcal{P}_{<} \subset \mathcal{P}_{s,t}$ for which for all $p \in \mathcal{P}_{>}$, $f^{U \! E}_p > \hat{f}^{SO}_p$, and for all $p' \in \mathcal{P}_{<}$, $f^{U \! E}_{p'} < \hat{f}^{SO}_{p'}$ and all paths only differ by links in $E'$.
This is because the total flow between any origin-destination is larger in the $SO$ flow by equation~\eqref{eq:demandcon}.
Moreover, $\sum_{p \in \mathcal{P}_{>}}(f^{U \! E}_p - \hat{f}^{SO}_p) \le \sum_{p' \in \mathcal{P}_{<}}(\hat{f}^{SO}_p - f^{U \! E}_p)$ since the flow along non-constant latency links constrains the total flow.
Move $\sum_{p \in \mathcal{P}_{>}}(f^{U \! E}_p - \hat{f}^{SO}_p)$ units of flow from paths in set $\mathcal{P}_{>}$ to paths in set $\mathcal{P}_{<}$ in the $SO$ flow $\hat{f}^{SO}$.
Denote the new flow by $f'$.
The total travel time for $f'$ cannot increase because the flow has only increased on constant latency links, and the new flow does not exceed $\bar{f}_e^{SO}$ on any link.
The total travel time also cannot have decreased because $\hat{f}^{SO}$ was an $SO$ flow, so $f'$ is also an $SO$ flow.
Continue this procedure until there does not exist a link $e \in E'$ for which $f^{U \! E}_e$ exceeds the transformed $SO$ flow.
Then we have constructed an $SO$ flow, $f$, for which for all links $e \in E$, $f^{U \! E}_e \le f_e$, a contradiction.
\end{proof}

\begin{lemma}
  \label{lemma:flow_minus_subflow}
  For a network $\{G,R\}$, let $f^*$ be a subflow of a feasible flow $f$. Then the flow $f'$ such that $f'_e = f_e - f^*_e$ is also a subflow of $f$.
\end{lemma}

\begin{proof}
  First, $0 \le f'_e \le f_e$, by the definition of a subflow.
  Now set $r'_{s,t} = R(s,t) - r^*_{s,t}$.
  Then for all $s,t \in V$, $\sum_{e \in out(s)} f'_e -\sum_{e \in in(s)} f'_e = \sum_t (R(s,t) - r^*_{s,t}) = \sum_t r'_{s,t}$, and similarly for $\sum_{e \in in(t)} f'_e -\sum_{e \in out(t)} f'_e$
\end{proof}

\begin{theorem}
  The optimal value of the $U \! E$ linear program for a network instance $\{G,R\}$ is the maximum amount of $U \! E$ agents that the network can support and achieve $SO$.
\end{theorem}

\begin{proof}
% \MA{Note that we don't need strictly increasing latency functions anymore, but this proof will only work with strictly convex $c(f)$ functions.} First, note that for strictly convex $c(f)$ functions, the flow on each link at $SO$ is unique \MA{Need reference here}.
First, by Theorem~\ref{theorem:UE_flow_is_subflow}, there exists an $SO$ flow such that the optimal $U \! E$ subflow, $f^{U \! E}$, is a subflow of the $SO$ flow, and by Lemma \ref{lemma:flow_minus_subflow}, there exists a subflow of compliant agents that can achieve the $SO$ solution.
Moreover, by the definition of the $U \! E$ linear program and Lemma~\ref{lemma:identical_reduced_cost_paths}, the $U \! E$ flow is only along zero reduced cost paths.
By the definition of zero reduced cost paths, all $U \! E$ agents are willing to take the assigned paths.
Therefore, the $SO$ solution is achievable with the $U \! E$ flow, and there is some volume of $U \! E$ flow that is equal to the objective of the $U \! E$ linear program.

Now, suppose that there was another $U \! E$ flow assignment, $f'$, for which compliant flow could be assigned in such a way that the $SO$ total system travel time was achieved and the total $U \! E$ flow volume was larger than the value returned by the $U \! E$ linear program.
Note that this flow assignment ($f'$) must be a subflow of some $SO$ flow, $f$.
Moreover, by the definition of $U \! E$ flow and the fact that all paths in a $SO$ solution are minimum marginal cost paths, all paths assigned with a $U \! E$ flow greater than zero must be a zero reduced cost path.
Therefore, the flow $f'$ satisfies the equations~\eqref{eq:demandcon}-\eqref{eq:nonnegative}, and since the $U \! E$ linear program returns the optimal $U \! E$ flow assignment, this is a contradiction.
\end{proof}

% \MA{Should also point out here that we can \emph{check} whether or not a UE flow will permit an SO compliant assignment by checking if it is a solution to the UE LP (not necessarily an optimal solution).
% Therefore, we can both compute the maximal amount of UE flow that can be tolerated, and we can determine whether or not a given flow will permit a UE flow, a bit stronger.}

While we've demonstrated that we can compute the \emph{maximal} $U \! E$ flow that permits an $SO$ solution given the appropriate assignment of the compliant flow, it is likely that a more common problem would be to determine, for a given set of compliant agents, whether or not it is possible to achieve $SO$ with that set.
Our methodology also provides an answer to this question, as the following Corollary demonstrates.

\begin{figure*}[tbp]
\centering
%\vspace{-0.2 cm}
\includegraphics[width=0.95\textwidth]{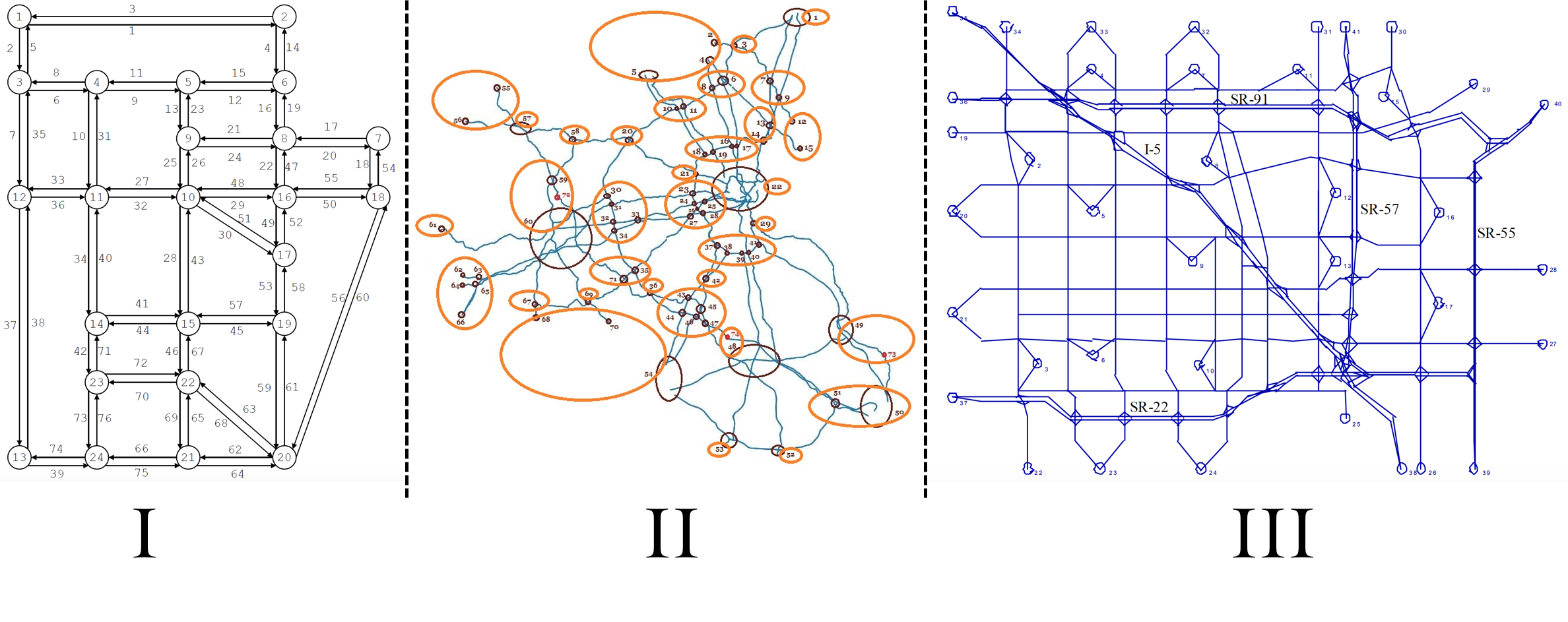}
\caption{Three representative network topologies: I - Sioux Falls, SD, II - Eastern Massachusetts (Ellipsoids represent different zones), III - Anaheim, CA.}%\vspace{-0.2 cm}
\label{fig:networks}
% \label{fig:badc}
\end{figure*}

\begin{corollary}
  \label{corollary:subflow_exists}
  For a given network instance $\{G,R\}$ and given a set of compliant demand, $r^C_{s,t}$, from each origin destination pair $s , \, t \in V$, there exists a compliant flow $f^C$ such that the network achieves $SO$ if and only if there exists an $x^s_e$ for all $s \in V$ and $e \in E$ such that $r^{U \! E}_{s,t} = R(s,t) - r^C_{s,t}$ and $x^s_e$ are a solution to the $U \! E$ linear program.
\end{corollary}

\begin{proof}
  By Theorem~\ref{theorem:UE_flow_is_subflow}, any solution to the $U \! E$ linear program defines a subflow of an $SO$ flow.
  Therefore, if $r^{UE}_{s,t}$ and $x^s_e$ is a solution, there exists an assignment of the compliant flow that achieves $SO$.

  Moreover, if there exists an assignment of the complaint flow, $f^C$, such that a $U \! E$ subflow with demands $r^{UE}_{s,t}$ achieves system optimum, then the $U \! E$ flow is only along zero reduced cost paths by definition of $U \! E$ flow and $SO$, and the $U \! E$ subflow is feasible.
  Therefore, the decomposed $U \! E$ flow satisfies the constraints of the linear program.
\end{proof}

\section{Flow Assignment for Compliant Agents}

% \MA{This section still needs work to argue that we actually compute the assignment of compliant users, as discussed in the email.}

Given that we can now determine both the maximal amount of $U \! E$ flow that a system can tolerate and achieve system optimum and, for a given set of compliant agents, whether or not a system can achieve optimum, we are only left with assigning the compliant flow to paths.
This section tackles the question of how to assign paths to a, sufficiently large, set of compliant agents such that $SO$ is achieved.

% Computing the maximal amount of self-interested demand is a first step towards reaching $SO$. However, achieving $SO$ requires us to also address the following:

% \begin{enumerate}
% \item Which subset of the demand should be assigned as compliant?
% \item How should the compliant flow be assigned to links?
% \end{enumerate}

% \MA{The following is no longer strictly true. We now don't know for sure which SO solution we've found because it's not necessarily unique. We'll need another linear program... Guni: Why does it matter which SO solution we've found? For the SO solution that we computed we can answer these two questions using the LP output. Michael: We don't know which SO solution we are a subflow of, so we can't just subtract the flow to get the remaining flow. The LP previously defined doesn't give us the total system flow, because SO flows are not unique when latency isn't strictly increasing.}
The methodology from the previous section immediately suggests a solution.
Given a network instance $\{G,R\}$, suppose that we have compliant demand equal to $r^C_{s,t}$ for all $s, \,t \in V$.
Then we must find a $SO$ flow, $f^{SO}$, such that $r^C_{s,t}$ and $r^{U \! E}_{s,t} = R(s,t) - r^C_{s,t}$ permit subflows of the $SO$ solution.
Such a $SO$ flow must exist by Theorem~\ref{theorem:UE_flow_is_subflow} and Corollary~\ref{corollary:subflow_exists}.

The first step is to compute the $U \! E$ subflow, $f^{U \! E}$, given $U \! E$ demand.
From the previous section: this exists and is computationally tractable.
Any feasible subflow, $f^C$, with demand $r^C_{s,t}$ such that the total flow along link $e$ satisfies $f_e^C+f_e^{U \! E} \le \bar{f_e}^{SO}$ has latency equal to the $SO$ solution, and the flow $f_e^C + f_e^{U\!E}$, by Lemma~\ref{lemma:constant_latency_SO}, is an $SO$ solution.

We can compute $f^C$ with the following linear program:
\begin{align*}
&\qquad \qquad \qquad \qquad \quad \max_{f^C_e} ~ 1&&  \nonumber \\
&\textit{subject to} && \nonumber\\
&\sum_{e \in out(v)} f^C_e - \sum_{e \in in(v)} f^C_e = \sum_{t} (r^C_{v,t}) &&  \forall \, v \in V \\ 
&\sum_{e \in in(v)} f^C_e - \sum_{e \in out(v)} f^C_e= \sum_{s} (r^C_{s,v}) &&  \forall \, v \in V \\
&0 \le f_e^C \le \bar{f}_e^{SO} - f^{U \! E}_e &&\forall \, e \in E 
% &r^*_{s,t} \leq R(s,t) &&  \forall \, s,t \in V \label{eq:demandcon} \\ 
% &\sum_{e \in out(s)} x^s_e  = \sum_{t \in V} r^*_{s,t}  && \forall \,s \in V  \label{eq:flowcon_s} \\ 
% &\sum_{e \in out(v)} \! x^s_e -  \sum_{e \in in(v)} \! x^s_e = -r^*_{s,v}  &&  \forall \, s,  v \ne s \in V \label{eq:flowcon_t} \\ 
% &\sum_{s} x_{e}^s \leq f^{SO}_e && \forall \, e \in E, \,  \, s \in V  \label{eq:bundle}\\
% &x_{e}^v \geq 0 ,~ r^*_{s,t} \geq 0  &&\forall \, v, s, t \in V, \, e \in E  \label{eq:nonnegative} \\
\end{align*}
We know that a solution to the above linear program exists and it can be computed tractably.

The final step is to decompose the compliant flow, $f^C$, into a per path assignment for each origin-destination pair $(s,t)$ in order to assign individual agents to a path.
This can be done in time $O(|V||E|)$ using standard flow decomposition algorithms (see Section 3.5 of Ahuja, Magnanti, et. al. (\citeyear{Ahuja.Magnanti.ea-1993}) for a discussion).

% The compliant flow can also be retrieved from the suggested LP. 
% Recall that variable $f^{U\!E}_e$ is assigned the $U \! E$ flow on link $e$ that originated from $s$. As a result the desired compliant flow, $f^C_e$, on link $e$ equals $f^{SO}_e - f_e^{U \! E}$.

% \section{Non-decreasing Latency Functions}
% \label{sec:non-decr-latency}

% \MA{Will eventually remove this section all together. Keeping it for the ability to take text from it if necessary.}

% While strictly increasing latency functions are common in theoretical models and, in this setting, leads to a simpler discussion, in realistic road networks latency functions will often be constant over a range of flows.
% This is likely to happen for links of the network that are far from capacity.
% Computing the optimal $U \! E$ assignment becomes significantly more difficult in this setting due to the $SO$ flow no longer being unique.
% Therefore, different $SO$ flows, while all achieving the same total system travel time, may permit different maximum amounts of $U \! E$ agents in the system.
% In this section, we modify the $U \! E$ linear program, allowing us to provably and tractably compute the maximal amount of $U \! E$ flow that the network can tolerate, and we provide a method to determine the assignment of compliant flow in order to achieve system optimum.

\section{Experimental Results}

\begin{table*}[t]
\centering
\begin{small}
\begin{tabular}{|l|r|r|r|r|r|r|r|r|r|}
  \multicolumn{1}{|c|}{Scenario}	& \multicolumn{1}{|c|}{Vertices}	&	\multicolumn{1}{|c|}{Links}	& \multicolumn{1}{|c|}{Zones}	&	\multicolumn{1}{|c|}{Total Flow}	 &	\multicolumn{1}{|c|}{$U \! E$ TTT}	&	\multicolumn{1}{|c|}{$SO$ TTT}	& \multicolumn{1}{|c|}{\% Improve} &	\multicolumn{1}{|c|}{Threshold}	&	\multicolumn{1}{|c|}{\% compliant}	\\
	\hline
Sioux Falls				&	24		&	76		&	24		&	360,600			&	7,480,225
		&	7,194,256	& 3.82	&	6.19E-11	&	\textbf{13.04}	\\
Eastern MA   &   74 		&   258 	&   74 		&   65,576   	&   28,181 		&   27,323  	& 3.04	&   3.04E-13 	&   \textbf{19.73} \\
Anaheim					&	416		&	914		&	38		&	104,694		&	1,419,913 &	1,395,015 & 1.75 &	8.05E-11	&	\textbf{19.76}	\\
Chicago S			&	933		&	2950	&	387		&	1,260,907	&	18,377,329
		&	17,953,267 & 2.31
		&	9.14E-10	&	\textbf{27.29}	\\
Philadelphia			&	13389	&	40003 	&	1525	&	18,503,872	&   335,647,106
	&	324,268,465 & 3.39
	&	4.20E-09 	&	\textbf{49.59}	\\
Chicago R		&	12982	&	39018	&	1790	&   1,360,427	&	33,656,964
		&	31,942,956 & 5.09
  	&	4.14E-07	&	\textbf{53.34}	\\
\end{tabular}
\end{small}
\caption{Required fraction of compliant agents given as ``\% compliant" for different scenarios along with network specifications for each scenario: number of vertices, links and zones followed by the Total Travel Time (TTT) at $U \! E$ (0\% compliant agents) and $SO$ (100\% compliant agents). The percentage of improvement of the $SO$ TTT over the $U \! E$ TTT is given as ``\% improve". }
\label{Tab:res}
\end{table*}

% \JH{If the motivation section is deleted then this first statement should be revisited to introduce road networks as a primary motivation for this work.}
We are interested in the viability of \emph{opt-in} micro-tolling schemes to more efficiently utilize road networks.
As such, we haven undertaken an empirical study to investigate the minimal amount of compliant flow required for $SO$ ($r^*_{U \! E}$) in six realistic traffic scenarios over actual road networks.
%Mainly by showing that real-life traffic scenarios can achieve $SO$ with a fairly low portion of compliant agents. Results are presented for three (?) benchmark traffic scenarios. 

\subsection{Scenarios}

%\begin{figure*}[t]
%\centering
%\begin{center}%\vspace{-0.2 cm}
%\includegraphics*[width=\textwidth]{Figures//results.png}
%\end{center}
%\caption{I: The Sioux-Falls network. II: The downtown Austin network. III: Macroscopic-model results showing average travel time as a function of selfless portion of the traffic for both scenarios IV:  Results from the mesoscopic model (CTM) presented in a pattern similar to that observed in III.}%\vspace{-0.2 cm}
%\label{fig:results}
%\end{figure*}

Each traffic scenario is defined by the following attributes:

\begin{enumerate}
\item The road network, $G(V,E)$, specifying the set of vertices and links where each link is affiliated with a length, capacity and speed limit. Networks are, following standard practice, partitioned into traffic analysis zones (TAZs) and each zone contains a node belonging to $V$ called the centroid. All traffic originating and terminating within the zone is assumed to enter and leave the network at centroids.

\item A trip table which specifies the traffic demand between pairs of centroids. The demand function $R$ between nodes other than centroids is set to zero. 
\end{enumerate}

The following benchmark scenarios were chosen both for their diversity of topology and traffic volume and their widespread use within the traffic literature: Sioux Falls, Eastern Massachusetts, Anaheim, Chicago Sketch, Philadelphia, and Chicago-regional. All traffic scenarios are available at: \url{https://github.com/bstabler/TransportationNetworks}. Figure~\ref{fig:networks} depicts three representative network topologies (the three smallest networks).

\subsection{The Traffic Model}
\label{sec:macro}

A macroscopic model was used in order to evaluate traffic formation.
Macroscopic models calculate the $U \! E$ in a given scenario using algorithm B~\cite{dial2006path}. For all scenarios, the model assumed that travel times follow the \textit{Bureau of Public Roads} (BPR) function~\cite{moses2017calibration} with the commonly used parameters $\beta = 4$, $\alpha = 0.15$. 
The $SO$ solution is computed by replacing the latency functions with $c_e'(x)$ and using algorithm B to obtain the equilibrium solution \cite{dial1999part1}.
Since solving for the $U \! E$ and $SO$ solutions requires solving a convex program \cite{dial2006path}, we only solve them to a certain precision.
To measure convergence, given an assignment of agents to paths, we define the average excess cost (AEC) as the average difference between the travel times on paths taken by the agents and their shortest alternative path.
The algorithm terminates when the AEC is less than 1E-12 minutes (except for Chicago-regional for which 1E-10 was used due to the size of the network).
Therefore, a minimum marginal cost path is only a minimum up to a threshold.

%To define the set of zero reduced cost links for the linear program, we compute a threshold $T$ for considering a link as zero reduced cost. That is, a link, e, is zero reduced cost with respect to origin, $s$, if there is a path, $p$, leading from $s$ to the head of $e$ such that $l_p(f^{SO}) \le optl_{s,e}(f^{SO}) + T$ and $c'_p(f^{SO}) \le optc'_{s,e}(f^{SO}) + T$ where $optl$ and $optc'$ represent respectively the optimal latency and marginal cost values over paths leading from $s$ to the head of $e$ at the $SO$ solution. $T$ was set as

A link $e$ is defined to be zero reduced cost with respect to $s$ if it carries flow originating at $s$ in the SO solution (i.e., the link belongs to a minimum marginal cost path) and if the difference between the least latency path that include $e$ and the least latency unrestricted path, both leading from $s$ to the head vertex of $e$, is less than a threshold $T$. %where the shortest path are computed using the $l_e(f_e^{SO})$ weights (i.e., the link is on a minimum latency path up to the threshold).
% \MA{I have no idea what the sentence below this means. The reduced cost just means it is on a minimum latency, minimum marginal cost path. The reduced cost isn't a number, so what does a ``maximum reduced cost'' mean?}
% \TR{Is this more clear?}

The threshold $T$ is defined as follows. for each origin $s$ and link $e$ we calculate the least marginal cost path ($c'$) leading from $s$ to the head vertex of $e$ at the $SO$ solution. We do this once while restricting the path to include $e$ and once without such restriction. The difference between these two values is stored and $T$ is set to be the maximum of these difference across all the links and origins in the network.

%in an $SO$ path from origin $s$, we calculate the difference between the shortest path to the head node of $e$, where we force $e$ to be a part of the path, and the shortest unrestricted path to the head node, where the shortest path uses the marginal cost of each link, $c'_e(f_e^{SO})$. $T$ is then defined as the maximum of these values across all the links and origins in the network.

%To define the set of zero reduced cost links for the linear program, we compute a threshold that measures how far travelers in the $SO$ solution from the minimum marginal cost paths using any given link.
%To this end, for each origin $s$ and link $e$ in an $SO$ path from origin $s$, we calculate the difference between the shortest path to the head node of $e$, where we force $e$ to be a part of the path, and the shortest unrestricted path to the head node, where the shortest path uses the marginal cost of each link, $c'_e(f_e^{SO})$. The threshold is then defined as the maximum of these values across all the links and origins in the network.

\subsection{Results}
\label{sec:res}

Table \ref{Tab:res} presents the percentage of flow that must be compliant in order to guarantee an $SO$ solution for six different traffic scenarios.
Each scenario is affiliated with the number of vertices, links, and zones comprising the affiliated road network as well as the number of trips that make up the affiliated demand.
%Total travel time in minutes is given for the self-interested agents (routed by the $U \! E$ controller) and compliant agents (routed by the $SO$ controller) of each scenario along with the threshold used to compute these values.

The columns ``$U \! E$ TTT" and ``$SO$ TTT" represent the total travel time (in minutes) over all agents for the case where 100\% of the agents are controlled by the $U \! E$ controller ($U \! E$ solution) and when 100\% of the agents are controlled by the $SO$ controller ($SO$ solution) respectively.
The percentage of improvement in total travel time between $U \! E$ TTT and $SO$ TTT is also shown under ``\% improve".

The percentage of required compliant flow (formally $r^*_{U \! E}/|R|$ where $|R|=\sum_{s,t} R(s,t)$) as computed by the $U \! E$ linear program (Definition~\ref{definition:UE_LP}) is presented for each scenario under ``\% compliant".\footnote{Statistical analysis for Table \ref{Tab:res} is not presented, as the macroscopic model is deterministic.}

The results suggest that as the size of the network (i.e., the number of nodes and vertices) increases, a greater fraction of compliant travelers are needed to ensure the network achieves system optimum.
This appears to be due to an increasing number of used paths at the $SO$ solution as the network size increases.
As the number of paths grow, the set of zero reduced cost paths grows more slowly, and, therefore, a higher percentage of compliant agents is required. 

\section{Summary}
This paper discussed a scenario where a set of agents traverse a congested network, and a centralized network manager optimizes the flow (minimizes total latency) using a set of compliant agents.
A methodology was presented for computing the minimal volume of traffic flow that needs to be compliant in order to reach a state of optimal traffic flow.
Moreover, the methodology extends to inferring which agents should be compliant and how exactly the compliant agents should be assigned to paths.
Experimental results demonstrate that the required percentage of agents that are compliant is relatively small (between 13\% and 53\%) for several realistic road networks.

% show that in some scenarios (Sioux Falls, SD) rerouting as low as 13\% of the traffic reduces the total travel time by 85\% \TR{Change this}.

Going forward, it would be worthwhile to explore the possibility of approximation algorithms for assigning compliant flow when the $U \! E$ demand is too large to achieve system optimum.
Given that the optimal solution to this problem is known to be NP-hard, an efficient approximation algorithm would be a useful tool as opt-in network routing systems are implemented.
Further, in order to limit the necessary opt-in incentives, there is work needed to develop systems that target particularly influential users to opt-in to these systems.

\subsection{Acknowledgements}
\label{sec:acknowledgements}

The authors would like to thank Josiah Hanna and Michael Levin for contributing useful comments and discussions in the course of this research.

A portion of this work has taken place in the Learning Agents Research
  Group (LARG) at the Artificial Intelligence Laboratory, The University
  of Texas at Austin.  LARG research is supported in part by grants from
  the National Science Foundation (CNS-1305287, IIS-1637736,
  IIS-1651089, IIS-1724157), The Texas Department of Transporation,
  Intel, Raytheon, and Lockheed Martin.  Peter Stone serves on the Board
  of Directors of Cogitai, Inc.  The terms of this arrangement have been
  reviewed and approved by the University of Texas at Austin in
  accordance with its policy on objectivity in research.

The authors would also like to acknowledge the support of the Data-Supported Transportation Operations \& Planning Center and the National Science Foundation under Grant No.\ 1254921.

%% The file named.bst is a bibliography style file for BibTeX 0.99c
\bibliographystyle{aaai}
\bibliography{traffic}

\end{document}